\documentclass[lettersize,conference]{IEEEtran}

\ifCLASSINFOpdf
  \usepackage[pdftex]{graphicx}
  \graphicspath{{./}{../pdf/}{../jpeg/}}
  \DeclareGraphicsExtensions{.pdf,.jpeg,.png}
  \usepackage[caption=false, font=footnotesize]{subfig}
\else
\fi

\usepackage{amsmath}

\usepackage[cmintegrals]{newtxmath}

\usepackage{algorithmicx}
\usepackage{algorithm}
\usepackage{algpseudocode}
\makeatletter
\newcommand{\removelatexerror}{\let\@latex@error\@gobble}
\makeatother

\newtheorem{theorem}{Theorem}

\usepackage{cite}

\usepackage{tabularx}
\usepackage{booktabs}

\usepackage{color,soul}
\setulcolor{red}
\sethlcolor{yellow}
\usepackage{xcolor}

\renewcommand\hl[1]{#1} 
\usepackage{makecell}

\usepackage{amsthm,newtxmath}
\newtheorem{lemma}[theorem]{Lemma}
\newtheorem{prop}{Proposition}
\newtheorem{corollary}{Corollary}

\begin{document}
\bstctlcite{IEEEexample:BSTcontrol}

\title{A Cross-Layer Descent Approach for Resilient Network Operations of Proliferated LEO Satellites}

\author{\IEEEauthorblockN{Peng Hu}
\IEEEauthorblockA{
\IEEEauthorblockA{\textit{National Research Council of Canada}
}
\IEEEauthorblockA{\textit{1200 Montreal Road. Ottawa, ON K1A 0R6. Canada}
}\\
}
\IEEEauthorblockA{ \vspace{-8pt}
\IEEEauthorblockA{\textit{The Cheriton School of Computer Science, University of Waterloo}
\IEEEauthorblockA{\textit{200 University Ave W., Waterloo, ON N2L 3G1. Canada}}}
}
}

\markboth{Journal of \LaTeX\ Class Files,~Vol.~00, No.~0, August~2022}%
{Shell \MakeLowercase{\textit{et al.}}: A Sample Article Using IEEEtran.cls for IEEE Journals}


\maketitle

\begin{abstract} 
With the proliferated low-Earth-orbit (LEO) satellites in mega-constellations, the future Internet will be able to reach any place on Earth, providing high-quality services to everyone. However, high-quality operations in terms of timeliness and resilience are lacking in the current solutions. This paper proposes a multi-layer networking approach called ``Cross-Layer Descent (CLD)''. Based on the proposed system model, principles, and measures, CLD can support foundational services such as telecommand (TC) transmissions for various network operation missions for LEO satellites compliant with the Consultative Committee for Space Data Systems (CCSDS) standards. The CLD approach enhances timing and resilience requirements using advanced communication payloads. From the simulation-based analysis, the proposed scheme outperforms other classical ones in resilience and latency for typical TC missions. The future work and conclusive remarks are discussed at the end.

\end{abstract}


\section{Introduction}
\IEEEPARstart{W}{ith} the increasing launches of advanced low-Earth-orbit (LEO) satellites, the proliferated space assets will become an essential part of our future Internet. However, the operation of the LEO satellite in large and mega-constellations faces challenges due to the dynamics and complexity of the LEO satellite networks. The traditional way of contacting an individual satellite through limited access opportunities from a ground satellite operations center for Telemetry, Tracking, and Control (TT\&C) missions cannot meet the timing requirements for a large constellation of LEO satellites ranging from hundreds to thousands of small satellites. For example, if telecommand messages from a ground station (GS) to a set of target LEO satellites for updating radio resource allocation or routing mechanisms cannot be delivered within a stringent deadline, ground users would suffer from a reduced quality of service or even network outage. A missed or lost contact with the target LEO satellite also reduces resilience and raises security threats to the network operation. Therefore, an enhanced satellite network operations (SatNetOps) approach is required to address these issues. 

In the space industry, a geostationary (GEO) satellite network for replaying LEO satellite messages has recently been considered. Thanks to the constant availability of GEO satellites to a ground station (GS), GEO satellites can provide a service to relay messages to LEO satellites. Inmarsat reportedly started a new service called Inter-satellite Data Relay System (IDRS) in late 2020 for this service with the new GEO fleet equipped with the Ka/L-band payloads to support this, aligned with a vision of a deeper integration of GEO and non-geostationary (NGSO) satellites. However, as GEO satellites stay at a high altitude, the communication between a GEO satellite and an LEO satellite or GS would face efficiency and timing challenges. There are no further details on the technical design or efficiency analysis available. 

The inter-satellite links (ISLs) within a satellite constellation have been considered to relay the messages to a destination satellite in the delay/disruption-tolerant network (DTN) \cite{Shi17, NASA2021}. In this context, a few satellites are considered relay nodes to relay messages from source to destination spacecraft. A recent view of using multi-layer networking (MLN) for space communications using ISLs has been mentioned in \cite{Nishiyama13, Pachler21} to potentially improve data traffic. Recent solutions of the MLN-based schemes aim to improve the communication capacity and global navigation satellite system (GNSS) performance \cite{Ge2021, Jiang20_MLSN, Racelis2019}. The MLN schemes for SatNetOps have been recently explored in our previous works \cite{Hu2022_SatNetOps, Hu2022_LATINCOM}. However, the schemes do not discuss the framework for using more than three layers of satellite networks, theoretical discussion and route optimization for the cross-layer mechanism in MLN-based SatNetOps methods. 

To address the challenges mentioned above, we propose a novel MLN-based approach called ``cross-layer descent'' (CLD) considering multiple NGSO and GEO satellite constellations. The main contributions of the paper are summarized below:
\begin{itemize}
    \item The proposed CLD approach can address the timing and resilience challenges of the proliferated LEO satellite networks.
    \item A new MLN model with proposed satellite roles and resilience measures is formulated for generic MLN-based SatNetOps schemes. This allows the integration of multiple layers in NGSO satellite networks for designing optimal MLN-based SatNetOps schemes.
    \item The CLD approach is evaluated considering typical commercial satellite constellations and real-world configurations and parameters, where the proposed approach has shown effectiveness with the existing methods.
\end{itemize}

The remainder of the paper is structured as follows. The related work is discussed in Section II; the system model is discussed in Section III; the proposed CLD approach is discussed in Section IV; the evaluation of the proposed approach using real-world and planned satellite constellations is performed in Section V; and the conclusion and future work are discussed in Section VI.

\section{Related Work}

The use of MLN considering the LEO, MEO, and GEO layers has been discussed in DTN. The inter-layer contact selection problem based on contact-graph routing (CGR) is discussed in \cite{Shi17}, where the so-called ``space-time graph'' consisting of a series of static graphs at discrete instants is used in the system model. However, the timing requirements, such as real-time communication, are not aligned with the DTN objectives. Apart from DTN, the discussion on MLN with the recent launches of LEO satellite constellations is focused on solving the challenges of capacity improvements with the consideration of limited radio/compute resources on satellites. A Q-learning-based capacity allocation scheme with the consideration of a 3-layer network consisting of the LEO, MEO, and GEO constellations are considered in \cite{Jiang20_MLSN}, which does not consider the GS access and SatNetOps tasks. Besides, the use of multiple LEO satellite constellations has been considered to enhance the performance of a GNSS \cite{Racelis2019}, where MEO and LEO constellations can help improve spatial diversity by providing enhanced navigation performance for areas with poor position fixes. Recently, the proposed MLN-based SatNetOps methods have shown significantly improved performance compared to the traditional operational approach based on contact opportunities between a GS and a target spacecraft. In our earlier work \cite{Hu2022_SatNetOps, Hu2022_LATINCOM}, the MLN-based SatNetOps methods considering LEO/MEO/GEO layers are discussed with feasibility analysis, algorithmic and experimental validation. However, a detailed analysis for an optimal SatNetOps scheme has not been addressed.

The space industry and agencies have seen real-world systems using GEO satellites to relay messages due to GEO satellites' consistent availability to lower-orbit space assets and ground facilities. Inmarsat has recently provided a GEO-LEO relay service aligned with its view of an ``ORCHESTRA'' system where the NGSO and GEO satellites can ultimately be integrated as a system. The new GEO relaying service leveraging its new GEO satellite fleet with compatible payloads with the modern LEO satellites started in late 2020. In addition, NASA's Tracking and Data Relay Satellite (TDRS) system \cite{NASA2021} has used GEO satellites with multiple antennas for Ku/Ka/S-band communications to relay information for over 25 space missions, including the International Space Station and the tracking for the launch of James Webb Space Telescope. However, TDRS may well guide the capability of satellite payloads for MLN communication, but it can hardly provide system-wide design direction as it has different purposes than a SatNetOps mission.

The Consultative Committee for Space Data Systems (CCSDS) is a major organization for satellite operations standards, where the fundamental telemetry (TM) and telecommand (TC) missions are covered. The CCSDS' recent issue of TC space data link protocol (TC-SDLP) \cite{CCSDS232.0-B-4} covers the typical service support for TC missions, i.e., sequence-controlled and expedited services for various operations missions, where both services use the same data transfer process but the Automatic Repeat Request (ARQ) based flow control is only available in sequence-controlled service. The timing rules and MLN are essential to the MLN-based schemes for SatNetOps in compliance with CCSDS TM/TC standards; however, they are not mentioned in the standards efforts. 

In summary, the theoretical analysis and schemes for optimal MLN-based SatNetOps solutions for proliferated LEO satellites using various layers have not been addressed in the literature. 

\section{System Model}

Considering the inconsistent use of the layer concept in the literature, we define a layer as a conceptual sphere in the entire satellite network under consideration where an orbital shell of a satellite with the same altitude is located. For example, we have two layers if an entire satellite network has one LEO satellite constellation with two orbital shells. An upper-layer satellite is at a higher altitude than a lower-layer satellite. Let us then formulate the MLN model in the subsequent sections.

\subsection{MLN Model}
Suppose all satellite nodes in the $u$-th satellite network layer are in $S_{u} \subseteq S$, where the entire set of satellites is denoted by $S$ in a three-dimensional (3D) space. Within $S_u$ all satellite nodes are in $S_{u}$, where $u \in \mathbb{Z}^{+}$, and the number of layers is $n_u$ and the number of satellites at the $u$th layer is $n_{S_{u}}$. The $i$-th satellite and the $u$-th layer is $s_{u}(i)$. At the layer $u$, the accessible neighbouring satellites to $s_{u}(i)$ are denoted by $\mathcal{A}_{u}(i)=\{ s_{u}(\kappa_{u}(i,1)), s_{u}(\kappa_{u}(i,2)), ..., s_{u}(\kappa_{u}(i,k))\}$, where $\kappa(\cdot)$ is the mapping function to get the satellite ID number, and $k$ is the number of neighbouring nodes. The value of $u$ is ordered sequentially by altitude from low to high. The satellite number is incremented by one from $u=1$. For example, the satellites in $S_{1}$ is at a lower altitude than those in $S_{2}$. Thus, cross-layer neighboring nodes of $s_{u}(i)$ at immediate lower-and higher-altitude layers are $\mathcal{A}'_{u-1}(i)$ and $\mathcal{A}'_{u+1}(i)$, if $ u \geq 1$ and $u+1 \leq n_u$.

In practice, for satellites at the same layer, a typical case for the satellite network design can ensure there are four neighbouring nodes (two intra-plane nodes and two inter-plane nodes) with bi-directional ISLs, where $k=4$ based on a minimum elevation $\omega_{u}$ used for all satellites at the $u$th layer. In this case, for $s_{u}(i)$,  $\kappa(i,j) = i \pm ((j \cdot (\sigma_{u}+1) - 1)), ~ j \in \{0,1\}$, then we have $ \mathcal{A}_{u}(i)=\{s_{u}(i-1), s_{u}(i+1), s_{u}(i-\sigma_{u}), s_{u}(i+\sigma_{u})\}$, where $\sigma_{u,v}$ is the number of satellites in an ordinal plane of a layer. The available links to cross-layer nodes depend on the configuration of a satellite. Due to the size, weight, and power (SWaP) considerations, $s_u(i)$ may be equipped with antennas pointing to the lower altitude. 

To facilitate our discussion, we make further assumptions on the links based on the latest developments \cite{Sedin20, Park19_NASA}. For the inter-satellite links (ISLs), we assume all satellites can perform ISL communication with either free-space optical (FSO) and radio-frequency (RF) payloads. The cross-layer ISL happens from upper-layer satellites to lower-altitude satellites, where the upper-layer satellites are based on a minimum elevation $\omega_0$. For the ground-satellite links, the RF link and RF/FSO hybrid link \hl{following a certain ratio} are considered. Specifically, the Ka-band RF and 1550 nm wavelength FSO links are considered. These basic assumptions can facilitate our later discussion.

\subsection{Path Length in the MLN scheme}
Since the path length is an important variable for the MLN-based scheme for SatNetOps, different path segments need to be considered. Although the overall path length from the GS to a target LEO satellite $D$ is specific to the routes at the layers used, we can generally view $D$ from the segments of the overall path and express $D$ in (\ref{Eq:path_length}) as:

\begin{equation}\label{Eq:path_length}
D   = \sum_{i=1}^{H} d(i) 
    = \sum_{m=1}^{M}\sum_{u \in L_u(m)}^{} d_u(m)
\end{equation}

where $d(i)$ is the $i$th hop length, and $d_u(m)$ denotes the length of the path segment $m$ at the layer $u$. $H$ is the hop count, $m \leq M$ is the number of path segments (e.g., $M=3$), and $L_{u}(m)$ is the set of layers involved in the $m$-th path segment. Based on (\ref{Eq:path_length}), we can also see to minimize $D$, we need to minimize each path segment $d_u(m)$.

\subsection{Latency Measure} 

For the timing performance measure, we use the standard latency components. Let us let the overall latency of a TC transfer mission be $T_D$, consisting of  the propagation delay, $T_{D_{pa}}$, transmission delay, $T_{D_{t}}$, processing delay, $T_{D_{pc}}$, and queuing delay, $T_{D_{q}}$. If we let the size of a TC packet be $M$, $T_{D_{q}}$ per satellite can be assumed to be a small constant, $\eta$. Similarly, $T_{D_{pc}}$ per satellite can be assumed to be a constant, $\varepsilon$, which can incorporate the factor of switching delay of a satellite payload. For each hop $i$, we let the data rate be $r(i)$ and $T_{D_t}=M/r(i)$, where $r(i)$ equals $r_{RF}$ or $r_{FSO}$, representing either a Ka-band or FSO link is used. Then the latency can be expressed as (\ref{Eq:latency}) \cite{Hu2022_SatNetOps}:

\begin{equation}\label{Eq:latency}
 T_D =  \sum^{n_h}_{i=1}{\left( \frac{d(i)}{c} + \frac{M}{r(i)} + \eta + \varepsilon \right)}
\end{equation}
where $n_h$ is the total hop count, and $c$ is the speed of light. 

\subsection{Resilience Measure}
A new resilience measure $\mathcal{R}$ is defined in (\ref{Eq:resilience}) with the consideration of reliability factor, where the hop count is considered as it has a negative impact on path reliability over multiple hops as found in our previous work \cite{Hu2022_SatNetOps}.

\begin{equation}\label{Eq:resilience}
\mathcal{R}=\frac{\sum_{t=1}^{T}\mathcal{Q}(g(i),t)}{n_h \cdot T}
\end{equation}
where $g(i)$ is the mapping function to provide the global ID of $s_u(i)$ as SSN over the entire multi-layer network, ${G(t)}$ denotes the satellites accessible by SC at time $t$. $\mathcal{Q}(i,t)$ is an indicator function denoting if there is a successful path from $s_u(i)$ as SSN to DSN on $S$ at $t$, where $t =\{1,2,..., T\}$, and $T \in  \mathbb{Z}^{+}$ is the total time instants of a SatNetOps mission. 

\section{Proposed Cross-Layer Descent Approach}
With the system model described above, we can discuss the proposed CLD approach in this section.

\subsection{Roles of Satellite Nodes}
To differentiate the responsibilities and actions among satellite nodes in an MLN SatNetOps scheme, we define the satellite roles as follows:
\begin{itemize}
    \item Source GS (SGS): the GS sending messages for SatNetOps. An SGS can be viewed as a SatNetOps Center (SC).
    \item Source Satellite Node (SSN): the immediate satellite accessed by SGS. SSN is to be determined through the access opportunities to SGS at time $t$.
    \item Destination Satellite Node (DSN): the target satellite node.    
    \item Target-layer Source Satellite Node (TSN): the source satellite node at the layer where DSN is at.
    \item Intermediate Satellite Node (ISN): an intermediate satellite node (SN) on a path segment between TSN and DSN. 
\end{itemize}

\subsection{The Proposed Approach}

With the defined roles of nodes, we can further examine the principles of path determination between SSN and DSN, where \textit{cross-layer} means two layers are used, where \textit{intra-layer} means within the same layer.

\begin{lemma}
An overall optimal path during a timeframe $\Delta T$ of an MLN scheme between SGS and DSN is determined by (1) the optimal routes for every cross-layer path segment, and (2) the optimal route for every intra-layer path segment. %

\end{lemma}

\begin{proof}
    The overall path consists of the path segments SGS-SSN, SSN-TSN, and TSN-DSN. Let us first examine these segments separately in a general case, where SSN, TSN, and DSN are not the same node. Then, SGS starts the route determination at time $t$, and SGS-SSN is determined by the availability of SSN to SGS at time $t+T_0$. The SSN-TSN path segment is determined by the cross-layer routes at time $t+T_0+T_1$, where one or more layers may be used. The TSN-DSN is determined by the intra-layer path from TSN to DSN at $t + T_0+T_1+T_2$. The entire process takes within $\Delta T$ = $\sum_{}^{}{T_k}$. In special cases where SSN, TSN and DSN have overlapped roles, we can derive the two cases. In the first case that SSN and TSN are the same node, SGS-TSN and TSN-DSN are used where $T_1=0$; under this case, a subtle case is when SSN, TSN and DSN are the same node, then SGS-DSN is used where $T_1=T_2=0$. In the second case that TSN and DSN are the same node, SGS-SSN and SSN-DSN are used with one or more layers where $T_2=0$. Summarizing these cases makes the proof.
\end{proof}

\begin{prop}
    In the case that the cross-layer path segment is used, if the intra-layer route is fixed, to determine an optimal cross-layer SSN-TSN route, $SSN=s_{u+1}(i)$ and the optimal $TSN=s_{u}(j)$, where  $s_{u}(j) \in \mathcal{A}'_{u}(i)$ and $sgn(\theta_1) = sgn(\theta_2) $, where $\theta_1$ and $\theta_2$ are directional angles from $s_{u}(j)$ and $DSN$, respectively, to the SSN-SC segment, and $sgn()$ is the sign function. 
    
\end{prop}
\begin{proof}
From Lemma1 the an optimal cross-layer route is dependent on the path SSN-TSN but independent of the intra-layer path segment. Suppose $s_u(j)$ is a TSN, we need to first ensure it is on the same plane with DSN on the NED coordinate of SSN. This requirement can be guaranteed by the same sign of the directional angles expressed in the trigonometric form, i.e., $\theta_1 = \arccos(\frac{x_1}{|x_3|})$ and $\theta_2 = \arccos(\frac{x_2}{|x_3|})$ , where $x_1=\overrightarrow{\rm SSN~TSN}$, $x_2=\overrightarrow{\rm SSN~DSN}$, $x_3=\overrightarrow{\rm SC~SSN}$. The set of satellite nodes meeting the requirement is $\mathcal{A}''_{u}(i) \subset \mathcal{A}'_{u}(i)$, where $s_u(j) \in \mathcal{A}''_{u}(i)$. This makes the proof. 
\end{proof}

Due to the small size of a TC message following TC-SDLP, which lead to a small value of $\Delta T$, during when the entire network structure can be considered \hl{quasi-static}. The route selection through the cross-layer path segment can be determined by the slant range and path distance among intra-layer and cross-layer nodes, as shown in Fig. \ref{fig:path_illustration} with \hl{Lemma 1}. In addition, Proposition 1 can be demonstrated in the snapshot shown in Fig. \ref{fig:path_illustration}, where $s_2(1)$ is an SSN, $s_1(1)$ is TSN, and $s_1(3)$ is a DSN. $s_1(1)$ and $s_1(2)$ are ISNs and $s_1(1)$ can only communicate with $s_1(2)$. In the occurrences when the SC cannot access $s_1(1)$ but can access $s_2(1)$, then $s_2(1)$ becomes the SSN. The cross-layer SSN-TSN route selection from Layer 2 to Layer 1 is determined by the slant range between the SSN and neighbouring nodes, where the closest one is $s_1(1)$, and we can see $\theta_1$ and $\theta_2$ have the same sign. Then, the remainder of the inter-layer path at Layer 1 from TSN, ISN to DSN. Suppose there is another neighbouring node of SSN, $s_1'(1)$, located on the ``left'' side of the SGS-SSN plane, then its $\theta_1 < 0$, while $\theta_2 > 0$, which will not be considered as a TSN.

From \hl{Proposition 1}, we can derive the following corollary:
\begin{corollary}
If TSN and DSN are the same node, $|\theta_1| = |\theta_2|$. If TSN and DSN are not the same, there are two cases: (a) when DSN $\in \mathcal{A}'_{u}(i)$, $|\theta_1| > |\theta_2|$; (b) when DSN $\notin \mathcal{A}'_{u}(i)$, $|\theta_1| < |\theta_2|$.
\end{corollary}
\begin{proof}
The proof is based on analyzing all possible cases. When TSN and DSN are the same node, $\overrightarrow{SSN~TSN}$ and $\overrightarrow{SSN~DSN}$ are the same segment, so $\theta_1 = \theta_2$. In the case when TSN and SSN are not the same node, if DSN is accessible by SSN (where TSN already meets this condition), i.e., DSN $\in \mathcal{A}'_{u}(i)$, then SSN-DSN should be used as the shortest path which result in $|\theta_1| > |\theta_2|$; if DSN is not accessible by SSN, i.e.,  DSN $\notin \mathcal{A}'_{u}(i)$, $|\theta_1| < |\theta_2|$.
\end{proof}

\begin{figure}[!h]
    \centering
    \includegraphics[width=0.8\linewidth]{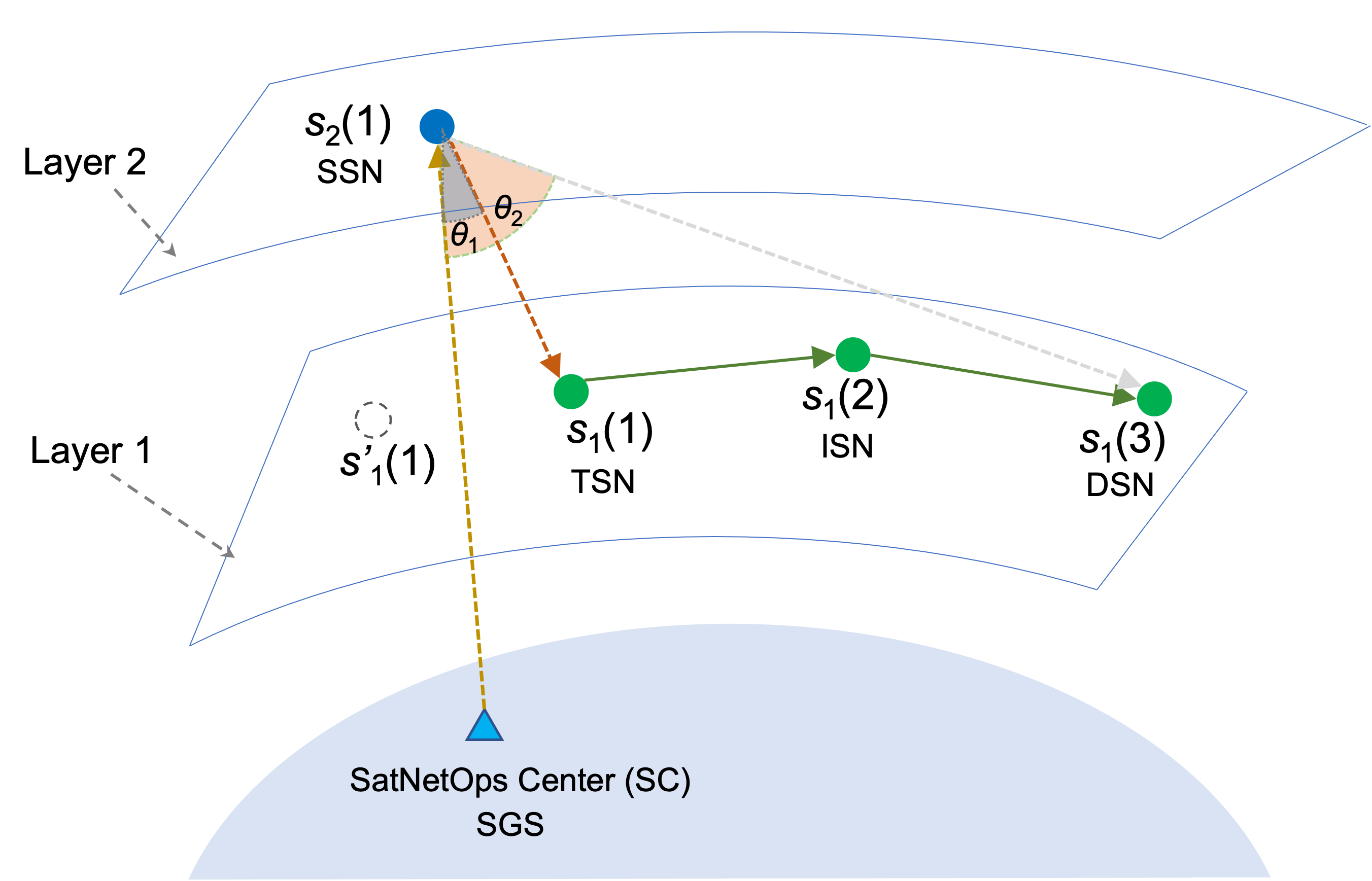}
    \caption{Geometry of the scenario snapshot for the CLD-based approach}
    \label{fig:path_illustration}
    \vspace{-5pt}
\end{figure}

Based on the discussion above, the CLD approach is described in Alg. \ref{Alg:cld_scheme}, where it consumes four input parameters and produces a cross-layer route vector in the return value. For clarity, we design a few custom helper functions in Alg. \ref{Alg:cld_scheme} with descriptive function names. For example, $\Call{IsAccessibleBySc}{ }$ determines if a satellite has an access to SC at time $t$, and $\Call{GetClosestSatToSc}{ }$ obtains the closes satellite to the SC at a layer. $SSN$ and $TSN$ can be determined dynamically through invoking $\Call{GetClosestSatToSc}{ }$ and $\Call{GetBestUnderLayerSat}{ }$, where the slant range is used to determine the best satellite in these functions. $\Call{GetIntraLayerRoute}{ }$ can calculate the route at the same layer based on the source and destination node, and $\Call{GetCrossLayerRoute}{ }$ calculates the cross-layer route based on \hl{Proposition 1} and \hl{Corollary 1}. The CLD scheme remains flexible to use different intra-layer routing schemes where the quasi-static status does alleviate the dynamic scheduling for changing routes. In the evaluation, we employ a regular routing scheme for $\Call{GetIntraLayerRoute}{ }$ where messages are forwarded through intra-plane neighbouring nodes first and then across the inter-plane nodes to the DSN.

Through the proposed approach, an SC can contact any target LEO satellite with the relay from current and upper layers. Based on (\ref{Eq:resilience}), for CLD approach, $\mathcal{Q}$ is defined in (\ref{Eq:Qindicator}):

\begin{equation}\label{Eq:Qindicator}
\begin{cases}
  \mathcal{Q}(g(i),t) = 0 & \text{for }|G(t)| = 0\\    
  \mathcal{Q}(g(i),t) = 1 & \text{for }|G(t)| > 0   
\end{cases}
\end{equation}
where 
$G(t)$ is the set of accessible satellites by SC at time $t$. 

\begin{algorithm}
\small
    \caption{CLD Approach}
    \label{Alg:cld_scheme}
    \begin{algorithmic}[1] 
        \Procedure{ Prepare}{$S$, $SC$, $n_u$, $dsnId$, $t$}
            \State $\rho \gets SC$
            \State $u \gets \Call{GetLayerIndexBySatId}{dsnId}$ \Comment{Set the starting layer index}
            \State $u_l(n_u) \gets \{0,...\}$  \Comment{Initialize a 1-D array of $n_u$ layers}
            \State $S_a \gets null $ \Comment{Initialize a list of sats accessible by SC}
            \For{ $S' \gets S(u)$}
                \State $i \gets \sum_{j=0}^{u-1}{n_{s_j} + 1}$
                \Comment{Set the initial sat ID}
                \For{$s \gets s_u(i),~i \leq n_{s_u}$}
                    \If{$\Call{IsAccessibleBySc}{\rho, s, t} == \text{TRUE}$}
                    \State $S_a \gets s_u(i)$ \Comment{Add $s_u(i)$ to $S_a$}
                    \State $u_l(u) \gets 1$
                    \EndIf
                    \State $i \gets i + 1$
                \EndFor
            \State $u \gets u + 1$
            \EndFor
            \State $u_{min} = \Call{GetMinLayerIndex}{u_l}$
            \If{$u_{min} == k$}
            \State SSN = \Call{GetClosestSatToSc}{$S_a$, $\rho$, $u_{min}$}
            \State $R1 \gets$ \Call{GetIntraLayerRoute}{$SSN$, $dsnId$}
            \Else
            \State $u_{ul} \gets$ \Call{GetUpperLayerIndex}{$u_l$}
            \State $SSN \gets$  \Call{GetClosestSatToSc}{$S_a$, $\rho$, $u_{ul}$}
            \State $TSN \gets$ \Call{GetBestUnderLayerSat}{$S_{u_{ul}}$, $SSN$}
            \State $R2 \gets$ \Call{GetCrossLayerRoute}{$SSN$, $TSN$}
            \State $R3 \gets$ \Call{GetIntraLayerRoute}{$TSN$, $dsnId$}
            \EndIf
        \State Return $\{R1, R2, R3\}$
        \EndProcedure
    \end{algorithmic}
\end{algorithm}

\section{Performance Evaluation}

\begin{table}[ht]
\centering
\setlength{\tabcolsep}{0.5em} 
\renewcommand{\arraystretch}{0.95}
\setlength{\tabcolsep}{0.8pt}
\caption{Key Simulation Parameters}
\label{Tbl:parameters}
\begin{tabular}{l l l}
\toprule
{\bfseries Parameter} & {\bfseries Value} & {\bfseries Notes}\\ 
\midrule
$n_{u}$ & 4 & Number of layers \\
$\{n_{S_1},n_{S_2},n_{S_2},n_{S_4}\}$ & $\{78,720,20,3\}$ & Num. of satellites in layers\\
$(lat, lng)$ & (51.0447, -114.0719) & SC coordinate\\
$gsMinElev$ & 25$^\circ$ & Min. elevation of SC\\
$\omega_0$ & 10$^\circ$ & Satellite min. elevation\\
$r_{FSO1}$ & 1.8 Gbps & FSO ground-space link rate\\
$r_{FSO2}$ & 10 Gbps & FSO ISL link rate\\
$r_{RF}$ & 324 Mbps \cite{Emily2018} & Ka-band RF link rate \\
$r_{HY}$ & 619.2 Mbps & RF/FSO1 hybrid link rate\\
$M$ & 1024 B & TC frame size\\
$T$ & 500 & 24hr-mission time samples\\
$T_{start}$ & 2022-09-01 01:00:00 & Mission start time (UTC) \\
$T_{stop}$ & 2022-09-02 01:00:00 & Mission stop time (UTC) \\
$T_{sample}$ & \hl{3600 s} & Sample time \\
$\varepsilon$ & 100 $\mu$s & Avg. processing delay\\
$\eta$ & 100 $\mu$s & Avg. queuing delay\\
Configuration I & $\{r_{RF}\}$ & RF space/ground links\\
Configuration II & $\{r_{FSO1},r_{FSO_2}\}$ & FSO space/ground links\\
Configuration III & $\{r_{HY}, r_{FSO_1}, r_{FSO_2}\}$ & Hybrid ground-space\\
& &  and FSO space-space links\\

\bottomrule
\end{tabular}
\vspace{-10pt}
\end{table}

\begin{figure}[!t]
\centering
{\includegraphics[width=\linewidth]{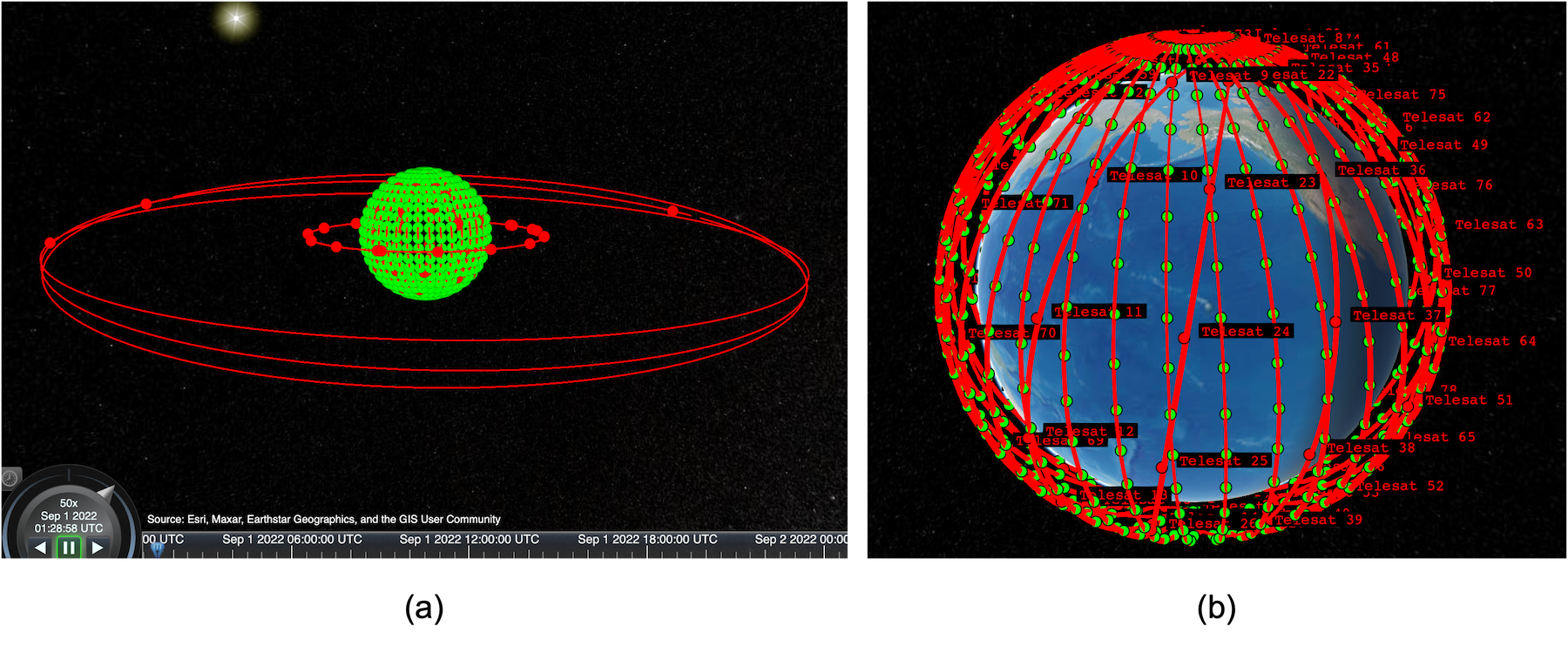}%
}
\caption{Satellite scenario showing (a) overall layers and (b) two LEO layers, where green and red dots represent OneWeb and Telesat Polar orbital satellites, respectively.}
\vspace{-10pt}
\label{Fig:satellite_scenario}
\end{figure}

\begin{figure}[!t]
\centering
\includegraphics[width=0.8\linewidth]{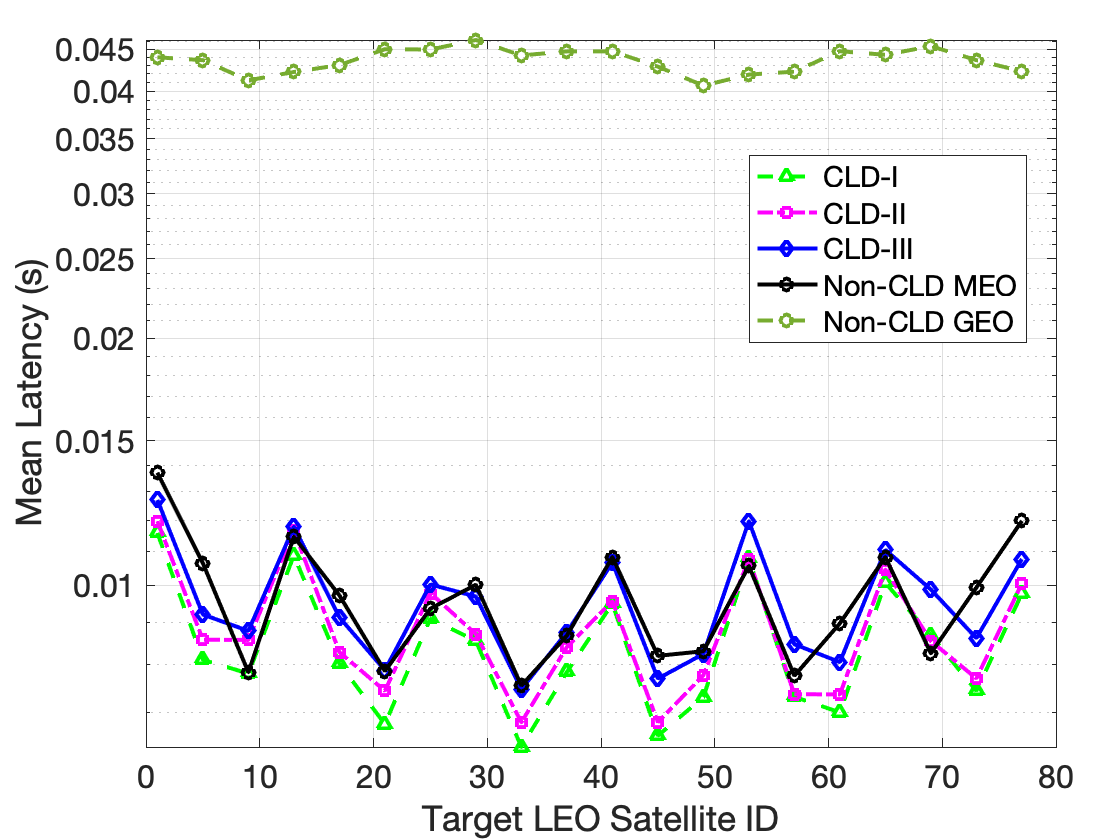}
\caption{Mean latency vs target LEO satellite ID in Configuration I}
\label{Fig:mean_latency_scenarios}
\vspace{-15pt}
\end{figure}

\begin{figure}[!t]
\centering
\includegraphics[width=0.8\linewidth]{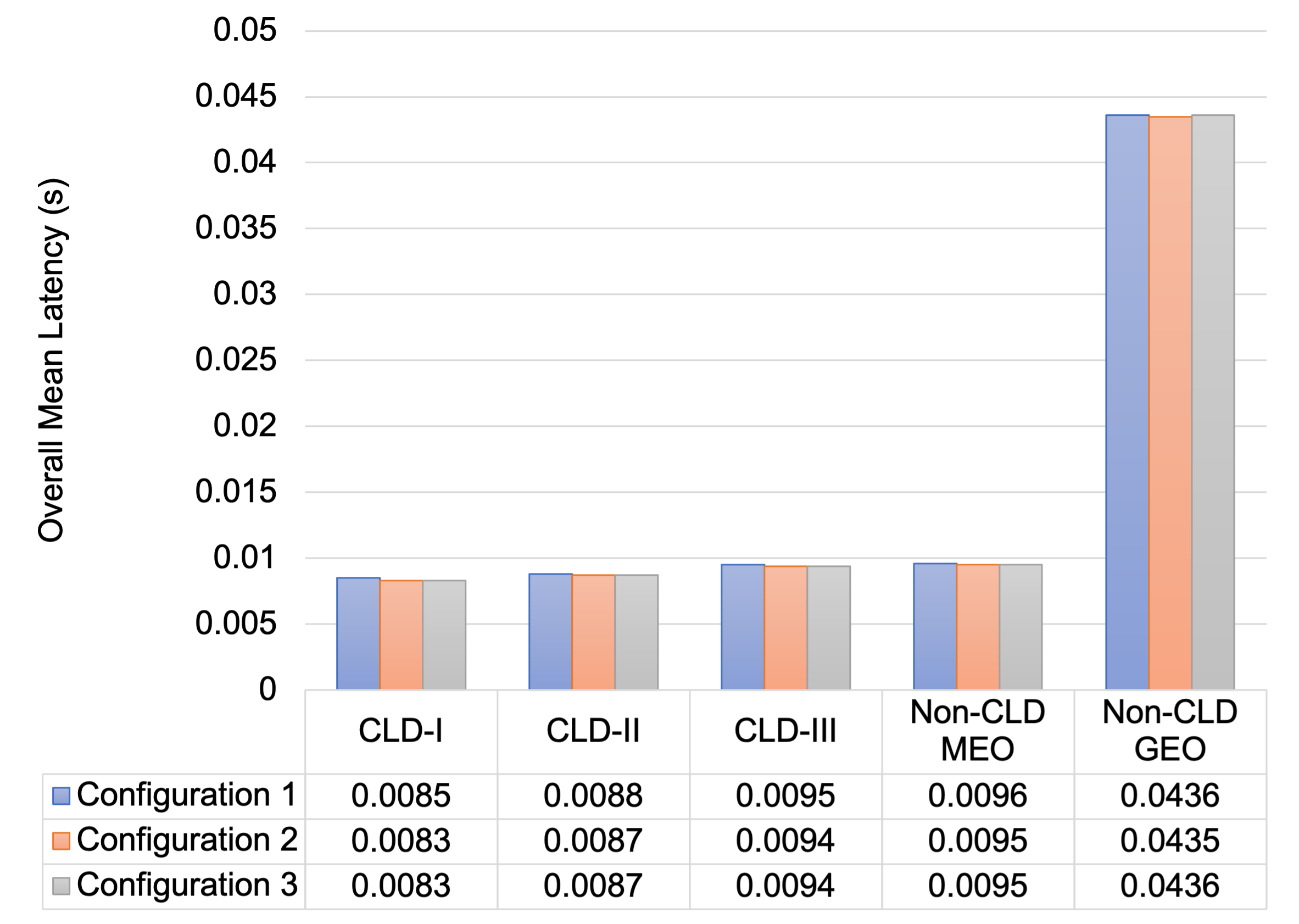}
\caption{Overall mean latency in Configurations I-III}
\label{Fig:overall_mean_latency}
\vspace{-15pt}
\end{figure}

\begin{figure}[!t]
\centering
\includegraphics[width=0.72\linewidth]{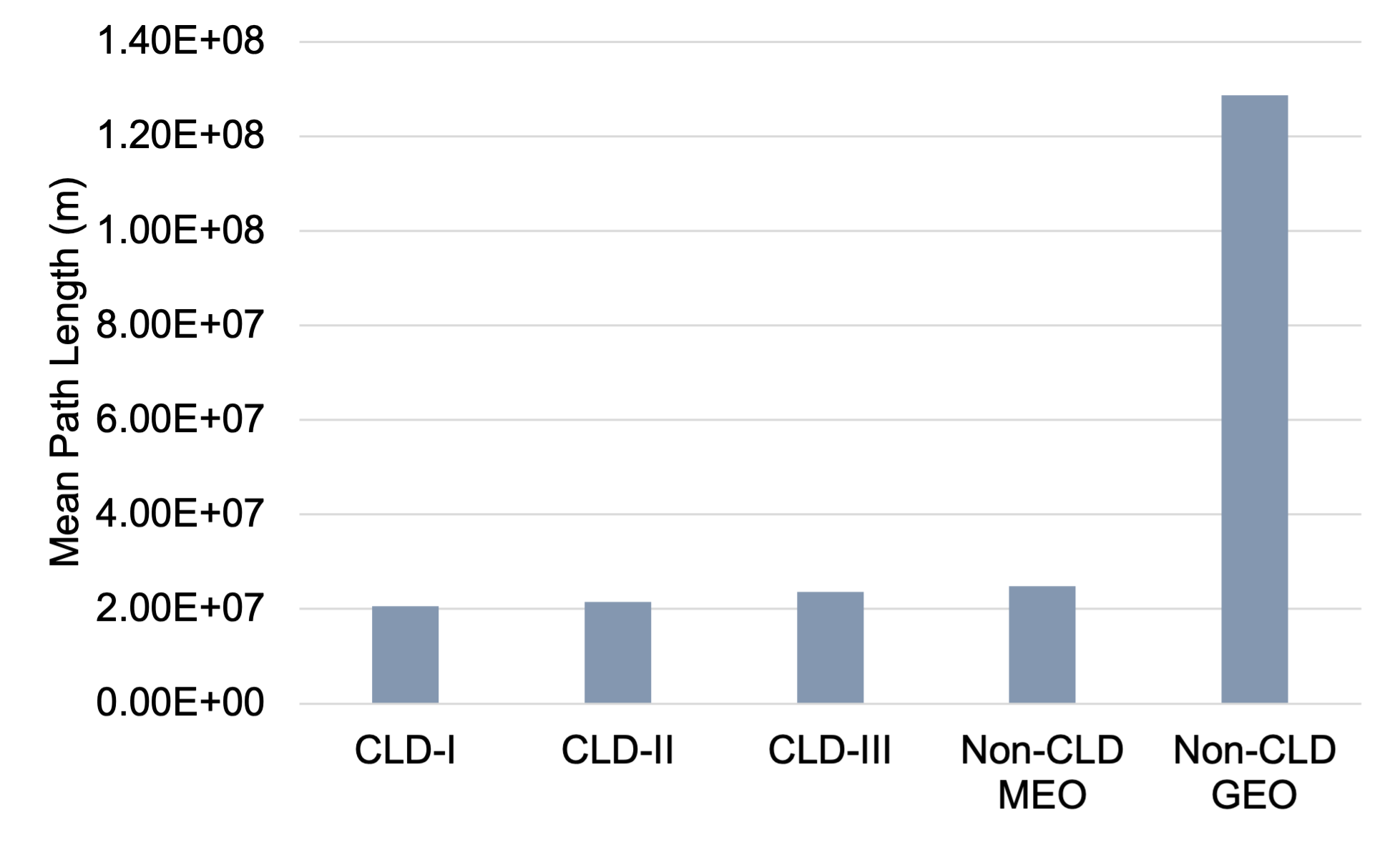}
\caption{Overall mean path length}
\label{Fig:overall_mean_path_len}
\vspace{-15pt}
\end{figure}

\begin{figure}[!t]
\centering
\includegraphics[width=0.72\linewidth]{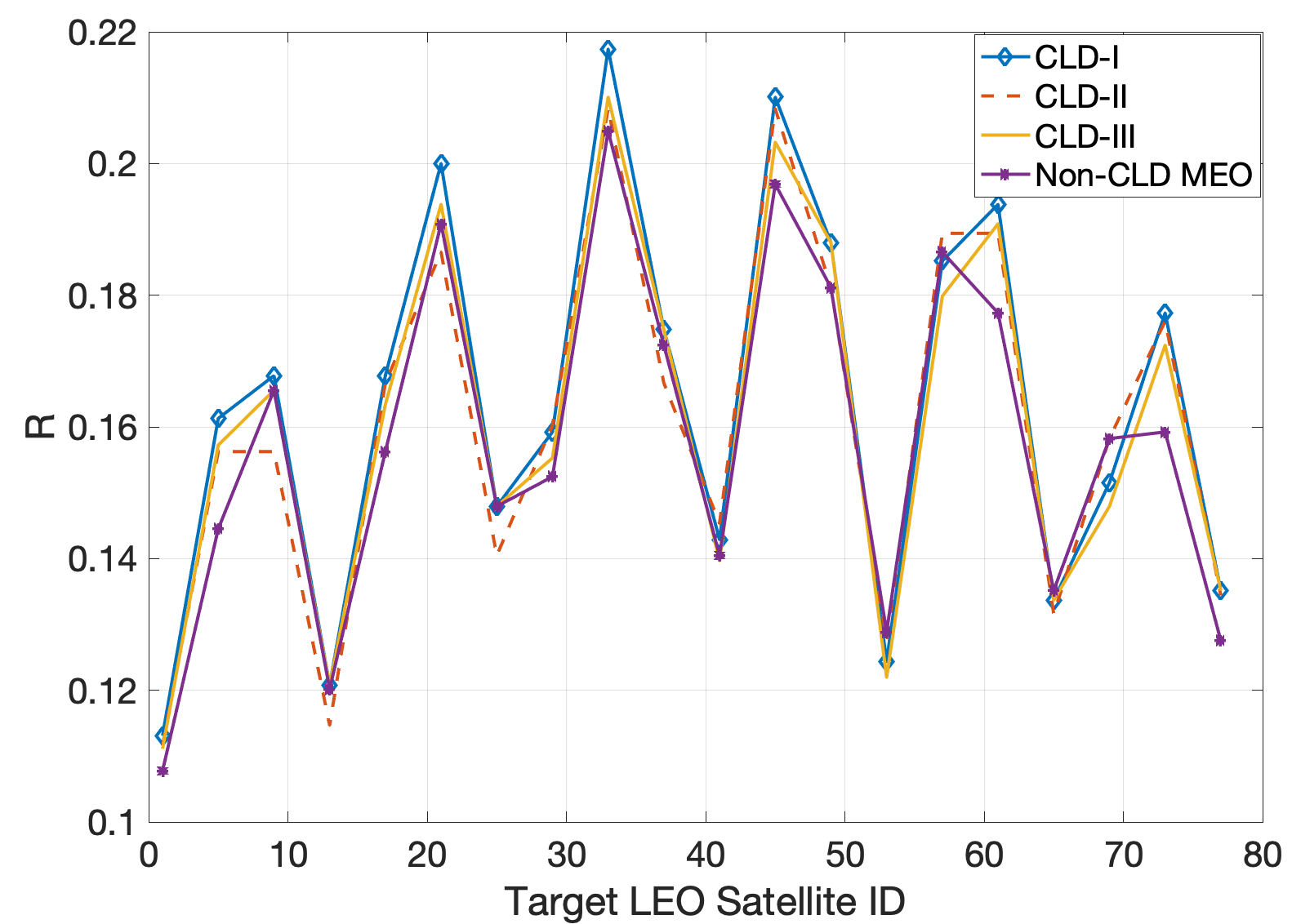}
\caption{Resilience performance}
\label{Fig:reslience_performance}
\vspace{-15pt}
\end{figure}

The MATLAB satellite communications toolbox is employed to validate the proposed schemes to achieve accuracy considering system parameters and configurations in simulations.
We consider typical commercial satellite constellations for $S$ in the satellite scenario as shown in Fig. \ref{Fig:satellite_scenario}. The scenario consists of two LEO satellite constellations (i.e., 78 Telesat Polar orbital and 720 OneWeb satellites), one MEO constellation (i.e., 20 SES O3b MEO satellites), and one GEO constellation (i.e., 3 Inmarsat-4 GEO satellites). The publicly available TLE files of O3b and Inmarsat-4 satellites are used, while the Telesat and OneWeb are based on the orbital propagator. All satellites are assigned global IDs from 1 to 821. The satellite mission parameters are shown in Table \ref{Tbl:parameters} based on our assumptions made in Section III. The target LEO satellites are in the Telesat Polar orbital shell as the first layer. To have good coverage of different target satellites across the orbital planes, we iterate the target satellite from 1 to 78 for every 4 satellites. To consider typical RF/FSO links used on space and ground entities, we consider three configurations, as shown in Table \ref{Tbl:parameters}, where a hybrid ground-space RF/FSO link with 20\% usage of RF is adopted in Configuration III.

Here we evaluate the proposed approach in comparison to the existing MLN-based SatNetOps schemes collectively referred to as ``Non-CLD MEO'' \cite{Hu2022_LATINCOM} and ``Non-CLD GEO'' \cite{Hu2022_SatNetOps}. We develop three schemes of the CLD approach to reflect the different scenarios when a variable number of layers is present: CLD-I considers all layers, CLD-II only uses the target LEO layer and MEO layer, while CLD-III only uses the target LEO layer and GEO layer. ``Non-CLD MEO'', where a target LEO layer and an MEO layer are used and ``Non-CLD GEO'', where only a GEO layer is used to contact a target LEO satellite.

Fig. \ref{Fig:mean_latency_scenarios} shows the latency performance per target satellite in Configuration I, where we can see the proposed CLD schemes have lower latency than non-CLD schemes. Compared to CLD-I, the ``Non-CLD GEO'' scheme has the worst latency performance by 5.13, 5.24, and 5.25 times in Configurations I-III, respectively. CLD-I can reach a latency as low as 6.35 ms for Configuration I (where the DSN ID is 33), where Ka-band links with $r_{RF}=384$ Mbps are used for all space-space and ground-space links. For Configuration II, where all FSO links are considered, the latency of CLD-I can be 6.22 ms. As the values in Y-axis for Configurations I-III are not easy to differentiate in the plots similar to Fig. \ref{Fig:mean_latency_scenarios}. Thus, Fig. \ref{Fig:overall_mean_latency} is plot to show the overall mean latency performance for each configuration. We can see that Configuration I has the longest latency than other configurations for each scheme. The ``non-CLD GEO'' scheme has the longest latency for all configurations. We can see CLD-I has the shortest latency, followed by CLD-II, CLD-III, and the ``non-CLD MEO'' scheme. The results in Configurations II and III are close as the variable $r$ on the ground-space link has a slight effect on the overall latency. The results indicate that using LEO satellite layers adjacent to the target layer creates shorter latency. It also shows if fewer layers are available, such as in the cases for CLD-II and CLD-III, the CLD approach still provides good timing performance.

The overall mean path length can help further explain the latency performance of different schemes over target LEO satellites. From Fig. \ref{Fig:overall_mean_path_len}, we can see that non-CLD schemes have the longest path length than CLD schemes. CLD-I has a mean path length of $2.06 \times 10^7$ m, while the traditional ``non-CLD GEO'' scheme has the longest path length of $1.29 \times 10^8$ m, which is 6.24, 5.99, and 5.45 times longer than those of CLD-I, CLD-II, and CLD-III, respectively. The results are mainly due to the effective cross-layer route selection, which can reduce the overall path length in CLD-based schemes.

Based on (\ref{Eq:resilience}), the resilience performance per target satellite is shown in Fig. \ref{Fig:reslience_performance}, where we can see CLD-I has the best resilience, followed by CLD-II, and CLD-III, while the ``Non-CLD MEO'' scheme has the lowest resilience for most target satellites. On average, the value of $\mathcal{R}$ of CLD-I is 3.74\% better than that of the ``Non-CLD MEO'' scheme. The results show that the CLD approach can provide better timing performance with improved resilience. Due to the worst latency performance of ``Non-CLD GEO'', its resilience performance is not compared here. 

\section{Conclusion}
The proposed CLD approach is presented in the paper, which shows effectiveness in the simulations in terms of latency and resilience compared to non-CLD methods. Furthermore, this work has derived an MLN system model, principles and new measures that formulate a foundation for designing and analyzing the MLN-based SatNetOps schemes. There is much room left for future contributions. The effect of various inter-layer routes will be explored in our future work.

 \section*{Acknowledgment}
We acknowledge the support of the High-Throughput and Secure Networks Challenge program of National Research Council Canada, and the support of the Natural Sciences and Engineering Research Council of Canada (NSERC), [funding reference number RGPIN-2022-03364].

\bibliographystyle{IEEEtran}

\bibliography{./bibitems}






\vfill

\end{document}